\documentclass[runningheads]{llncs}

\usepackage{amsfonts}
\usepackage{booktabs}
\usepackage{siunitx}
\usepackage{multirow}
\usepackage{graphicx}
\usepackage{listings}
\usepackage{commath}
\usepackage{epstopdf}
\usepackage{hyperref}
\usepackage{xcolor}
\usepackage{url}
\usepackage{mathtools}
\usepackage{alltt}
\usepackage{mathrsfs}
\usepackage{float}
\usepackage{subcaption}

\def\one{\mathds{1}}

\usepackage{graphicx,fancyvrb}
\usepackage[linesnumbered,ruled,vlined]{algorithm2e}
\usepackage{algpseudocode}
\usepackage{amsmath,amssymb}
\usepackage{booktabs}
\usepackage{caption}
\usepackage{float}
\usepackage{capt-of}
\usepackage{xcolor}

\usepackage{array}
\usepackage{arydshln}
\definecolor{mygreen}{rgb}{0,0.6,0}
\definecolor{mygray}{rgb}{0.5,0.5,0.5}
\definecolor{mymauve}{rgb}{0.58,0,0.82}

\usepackage{listings}
\lstnewenvironment{VerbatimText}[1][]{
    
    \lstset{fancyvrb=true,basicstyle=\footnotesize,captionpos=b,xleftmargin=2em,#1}
}{}

\newcommand{\ate}{{\tt \textsc{ATE}}}

\newcommand{\E}{{\tt \mathbb{E}}}
\newcommand{\pr}{P}

\newcommand{\bigCI}{\mathrel{\text{\scalebox{1.07}{$\perp\mkern-10mu\perp$}}}}

\newcommand{\mc}[1]{\mathcal{#1}}

\newcommand{\ignore}[1]{}

\newcommand{\indep}{\mbox{$\perp\!\!\!\perp$}}

\newcommand{\att}{\mathbf{A}}
\newcommand{\sx}{\mathbf{x}}
\newcommand{\bx}{\mathbf{X}}

\newcommand*{\rom}[1]{\expandafter\@slowromancap\romannumeral #1@}
\newcommand{\RNum}[1]{\uppercase\expandafter{\romannumeral #1\relax}}

\newcommand{\mb}[1]{{\mathbf{#1}}}

\newcommand{\proj}[1]{{\Pi}}
\newcommand{\sel}[1]{{\sigma}}
\DeclarePairedDelimiterX{\infdivx}[2]{(}{)}{%
	#1\;\delimsize\|\;#2%
}

\newcommand{\cut}[1]{}
\newcommand{\eat}[1]{}

\usepackage[utf8]{inputenc}
\usepackage{booktabs}
\usepackage{color}
\usepackage{dsfont}

\usepackage{amsthm}

\author{Luke Rodriguez\inst{1} \and Babak Salimi \inst{2} \and
Haoyue Ping \inst{3} \and\\ Julia Stoyanovich \inst{3}\thanks{This work
was supported in part by NSF Grant No. 1741047.} \and Bill Howe \inst{1} }

\authorrunning{L. Rodriguez et al.}

\title{MobilityMirror: Bias-Adjusted \\
Transportation Datasets}

\institute{Information School, University of Washington, Seattle WA\\
\email{\{rodriglr,billhowe\}@uw.edu}\\ 
\and
Computer Science and Engineering, University of Washington, Seattle WA\\
\email{bsalimi@cs.washington.edu}\\ 
\and
College of Computing \& Informatics, Drexel University, Philadelphia PA \\
\email{\{hp354,stoyanovich\}@drexel.edu}}

\date{}

\begin{document}

\maketitle 
\begin{abstract}
We describe customized synthetic datasets for publishing mobility data. Companies are providing new transportation modalities, and their data is of high value for integrative transportation research, policy enforcement, and public accountability.  However, these companies are disincentivized from sharing data not only to protect the privacy of individuals (drivers and/or passengers), but also to protect their own competitive advantage.  Moreover, demographic biases arising from how the services are delivered may be amplified if released data is used in other contexts.

We describe a model and algorithm for releasing origin-destination histograms that removes selected biases in the data using causality-based methods.  We compute the origin-destination histogram of the original dataset then adjust the counts to remove undesirable causal relationships that can lead to discrimination or violate contractual obligations with data owners.  We evaluate the utility of the algorithm on real data from a dockless bike share program in Seattle and taxi data in New York, and show that these adjusted transportation datasets can retain utility while removing bias in the underlying data.
\end{abstract}

\section{Introduction}

Urban transportation continues to involve new modalities including rideshare \cite{shuo:15}, bike shares \cite{zhang2016expanding}, prediction apps for public transportation \cite{ferris2010onebusaway}, and routing apps for non-motorized traffic \cite{brock2018sig}.  These new services require sharing data between companies, universities, and city agencies to enforce permits, enable integrative models of demand and ridership, and ensure transparency.  But releasing data publicly via open data portals is untenable in many situations: corporate data is encumbered with contractual obligations to protect competitive advantage, datasets may exhibit biases that can reinforce discrimination~\cite{NBERw22776} or damage the accuracy of models trained using them~\cite{mcfarland2015big}, and all transportation data is inherently sensitive with respect to privacy~\cite{demontjoye2013unique}. To enable data sharing in these sensitive situations, we advocate releasing ``algorithmically adjusted'' datasets that \textit{destroy causal relationships between certain sensitive variables} while \textit{preserving relationships in all other cases}.  

For example, early deployments of transportation services may favor wealthy neighborhoods, inadvertently discriminating along racial lines due to the historical influence of segregation~\cite{amazonrace}.  Releasing data ``as is'' would complicate efforts to develop fair and accurate models of rider demand.  For example, card swipe data for public transportation use in Seattle is biased toward employees of tech companies and other large organizations, while other neighborhoods typically use cash. This bias correlates with race and income, potentially reinforcing social inequities. Additionally, the privacy concerns of releasing this kind of data in raw form require careful attention.

Our focus in this paper is to model how these biased effects manifest in the context of transportation and how to correct for them in the context of individual privacy. We will consider three applications: ride hailing services (using synthetic data), taxi services (using public open data), and dockless bike share services (using sensitive closed data). 

We focus on dockless bikeshare services as a running example.
The City of Seattle began a pilot program for dockless bikes in the Summer of 2017, issuing permits for three different companies to compete in the area (Company A, B, and C).  To ensure compliance with the permits, these three companies are required to share data through a third-party university service to enable integrative transportation applications while protecting privacy and ensuring equity.  As part of this project, the service produces synthetic datasets intended to balance the competing interests of utility, privacy, and equity. 
Figure \ref{fig:seattle_gender} shows a map of the ridership for the pilot program in Seattle and is indicative of the kind of data products generated for transparency and accountability reasons.

There are several potential \emph{sensitive causal dependencies} in these datasets:
\begin{itemize}
\item Company A may be moving their bikes into particular neighborhoods to encourage commutes; this strategy could be easily copied at the cost of competitive advantage.
\item Company B may be marketing to male riders through magazine ads, leading to a male bias in ridership that could be misinterpreted as demand.
\item Company C may be negotiating with the city for subsidies for rides in underserved neighborhoods; they may be disallowed from publicly disclosing information about these subsidies, and therefore wish to remove the relationship between company and demographics. 
\item Ride hailing and taxi services allow passengers to rate and tip the drivers; gender or racial patterns in tips or ratings may encourage discrimination by drivers and should be eliminated before attempting to develop economic models of tip revenue.
\end{itemize}

\begin{figure}
    \includegraphics[width=2.5in]{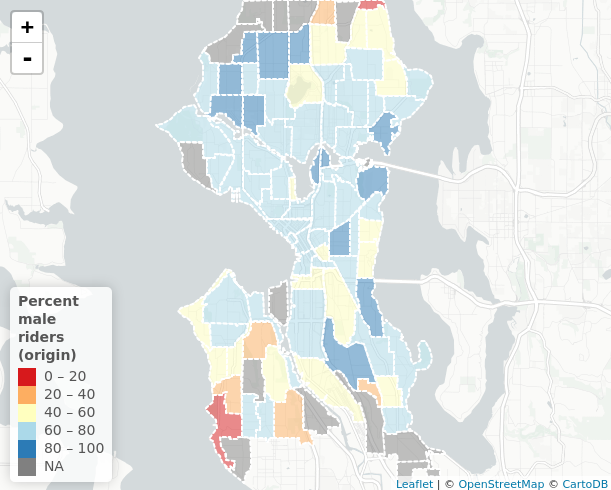}
	\includegraphics[width=2.5in]{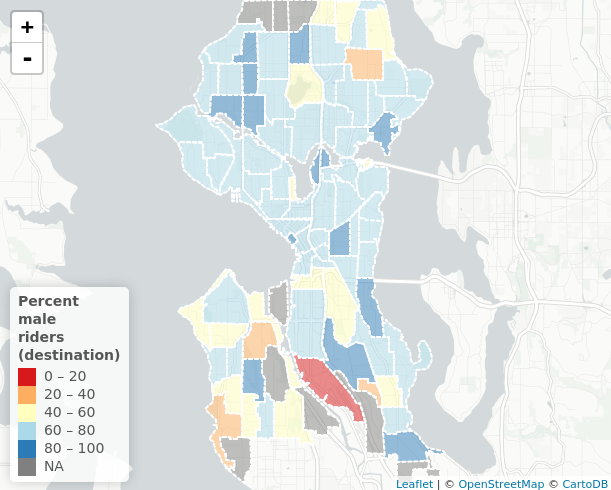}
    \caption{Percentage of bikeshare trips in Seattle with male riders by origin and destination neighborhoods}
    \label{fig:seattle_gender}
\end{figure}




In this paper, we develop an approach for adjusting transportation datasets to remove the effects of these sensitive causal relationships while preserving utility for classification and other analysis tasks. 

Transportation data is frequently released as an \emph{Origin-Destination} (OD) dataset: a set of location pairs representing city blocks, neighborhoods, or other spatial aggregation unit along with the traffic flow between the pair of locations.  We augment OD datasets with metadata, such that each tuple represents a histogram bucket corresponding to a unique combination of attributes. For example, the bike share data includes an attribute $gender$ with domain $(male, female, other)$ and an attribute $company$ with domain $(A, B, C)$ in addition to $origin$ and $destination$ attributes, each with a domain of 90 neighborhoods around Seattle.  A released dataset then might include the tuple $(female, A,$ $ Downtown, Ballard, 245)$ indicating that there were 245 trips taken by female riders on bikes owned by company B from Downtown to Ballard during the time period covered by the dataset.  These generalized OD datasets are sufficient for a variety of analytics tasks, including modeling demand, evaluating equity, estimating revenue, analyzing ridership trends, and estimating the effect on parking and motorized traffic.   

Although these datasets are aggregated, they can still expose sensitive information.  Individual privacy is an important concern in data sharing, but we do not focus on it here. In this work, we are interested in other forms of sensitive information encoded in the joint distribution across attributes.  To remove these sensitive patterns, the data publisher specifies a causal relationship between two attributes that they wish to eliminate in the adjusted dataset, conditioned on another set of attributes $Z$.  Then the causal repair problem is to set the mutual information between $X$ and $Y$ to zero, conditioned on $Z$.  
The conditional attributes $Z$ are important to express realistically complex situations; without these attributes, degenerate solutions such as scrambling or removing the $X$ or $Y$ attribute altogether would be sufficient.  

In our transportation context, our approach corresponds to computing a new distribution of trips over the buckets, one that preserves certain conditional joint probabilities while making other joint probabilities independent.
We apply this approach to two real-world datasets of interest: the NYC taxi trip record dataset \cite{tlc} and dockless bikeshare data from the city of Seattle.  The NYC taxi dataset is already available; we choose to evaluate on this dataset to ensure reproducibility. The bikeshare data is legally encumbered and cannot be shared publicly.

To evaluate the efficacy of our bias-reduction approach
we show that the distance between the original data and the adjusted data, as measured by multiple appropriate distance metrics, is no greater than would be expected due to sampling variance.  

Since our ultimate goal is to be able to release the datasets that we produce, we also investigate how we can adapt existing techniques from the differential privacy literature to work alongside our bias reduction approach. This requires that we carefully consider the domain of the data in order to produce results that have meaningful privacy guarantees. Moving beyond the theory, we also provide results showing how privacy preservation and bias reduction interact with each other in practice.

We make the following contributions:
\begin{itemize}
\item We describe the bias repair problem for transportation data, which arose from collaborations with companies and city agencies interested in sharing sensitive transportation data.
\item We describe a solution for removing a causal dependency (as defined by conditional mutual information) between two attributes in the context of transportation data. 
\item We evaluate this method on a synthetic rideshare dataset, a real taxi dataset, and a real bikeshare dataset, and demonstrate its effectiveness. 
\item We discuss generalizations of this approach to other domains, as we well as new potential algorithms to handle specific cases.
\item We evaluate how a carefully designed privacy-preserving algorithm that can be used in conjunction with our first solution to preserve utility while formalizing privacy protections.
\end{itemize}


The rest of this paper is organized as follows: in Section \ref{sec:rel} we describe related work in data sharing, causal analysis, and transportation. In Section \ref{sec:method} we present problem model and our proposed algorithm.  We describe taxi and bike sharing applications in In Section~\ref{sec:app}, and in Section \ref{sec:experiments} we evaluate the algorithm on real and synthetic data. We then extend these results in section Section \ref{sec:priv} by presenting an algorithm for preserving privacy along with an emperical evaluation of its performance. We conclude and discuss possible extensions in Section \ref{sec:conc}.

\section{Related Work}

\label{sec:rel}


Recent reports on data-driven decision making 
underscore that fairness and equitable treatment of individuals and groups is difficult to 
achieve~\cite{propublica,barocaseselbst,allegheny}, and that transparency and accountability of algorithmic processes are indispensable but rarely enacted~\cite{BrauneisGoodman,DBLP:journals/popets/DattaTD15,DBLP:journals/cacm/Sweeney13}.  
Our approach combines theoretical work relating causality to fairness~\cite{kilbertus2017avoiding} with practical tools for pre-processing data.

Recent research considers fairness, accountability and transparency properties of specific algorithms and their outputs. Dwork et al. articulated the fairness problem, emphasizing individual fairness (similar individuals receive similar outcomes), and Zemel et al. presented a method for learning fair representations based on this model that suppress discriminatory relationships while preserving other relationships~\cite{DBLP:conf/icml/ZemelWSPD13}. Feldman et al. provided a formalization of the legal concept of disparate impact~\cite{feldman:15}. Zliobaite presented a survey of 30+ fairness measures in the literature~\cite{DBLP:journals/datamine/Zliobaite17}.  However, these approaches are limited by the assumption that no information and no intervention methods are available for the upstream process that generated the input data~\cite{Kirkpatrick:2017:AD:3042068.3022181}.  Our focus is on developing a practical methodology that improves fairness for these upstream processes, specifically biased transportation data.

A common class of approaches to interrogate fairness and quantify discrimination is to use an associative (rather than a causal) relationship between a protected attribute  and an outcome.  One issue with these approaches is that they do not give intuitive results when the protected attribute exhibits spurious correlations with the outcome via a set of covariates.  For instance, in  1973, UC Berkeley was sued for discrimination against females in graduate school admissions,when it was found that 34.6\% of women were admitted in 1973 as opposed to 44.3\% of men.  
However, it turned out that women tended to apply to departments with lower overall acceptance rates; the admission rates for men and women when conditioned on department was approximately equal~\cite{salimi2018bias}. The data could therefore not be considered evidence for gender-based discrimination. 

The importance of causality in reasoning about discrimination is recognized in recent work.   Kusner articulated the link between counterfactual reasoning and fairness \cite{kusner2017counterfactual}.  Datta et al. introduce quantitative input influence measures that incorporate causality for algorithmic transparency to address correlated attributes \cite{DBLP:conf/sp/DattaSZ16}.  Galhotra et al. use a causal framework to develop a software testing framework for fairness~\cite{DBLP:conf/sigsoft/GalhotraBM17}. Kilbertus et al. formalize a causal framework for fairness that is closely related to ours, but do not present an implementation or experimental evaluation~\cite{kilbertus2017avoiding}. Nabi and Shpitser use causal pathways and counterfactuals to reason about discrimination, use causality to generalize previous proposals for fair inference, and propose an optimization problem that recomputes the joint distribution to minimize KL-divergence under bounded constraints on discrimination \cite{nabi2018fair}.  However, they do not provide an experimental evaluation, and do not propose an algorithm to eliminate causal relationships altogether.  No prior work uses these frameworks to generate synthetic data.  In our work, we focus on discrimination through total and direct effect of a sensitive attribute on an outcome. A comprehensive treatment of discrimination through causality requires reasoning about {\em path-specific causality} \cite{nabi2018fair}, which is difficult to measure in practice, and is the subject of our future work.

Prior work on publishing differentially private histograms, an intermediate step in generating synthetic data, was summarized by Meng et al \cite{meng2017different}.  All previous approaches assume a known, fixed domain, and look to improve utility over the  basic approach for numeric data proposed by Dwork et al. \cite{dwork2006calibrating}.

The first family of extensions are those that use hierarchical histogram structures. Xiao et al. propose a technique for using subcube histograms to improve accuracy in which the inputs are already binned into ranges \cite{xiao2012dpcube}. Similarly, the concept of universal histograms helps Hay et al. \cite{hay2010boosting} improve on the accuracy of Dwork's simple technique. Even so, they take as input a domain tree of unit intervals for the construction of the universal histogram. Building on the intuition that histograms depend heavily on bin choice, NoiseFirst and StructureFirst explicitly address both issues (Xu et al. \cite{xu2013differentially}).

More recently there have been a few generalized techniques proposed that are adaptable to the case of histogram publication. Privelet \cite{xiao2011differential} uses wavelet transforms and takes as input a frequency matrix of the counts to be approximated with noise proportional to the log of the number of such counts. Rastogi and Nath established the Fourier Perturbation Algorithm (FPA) \cite{rastogi2010differentially}, but in this framework the queries must explicitly define the domain before the algorithm is run. Building on this work, Acs et al. focus on the histogram problem and propose an extension of FPA called EFPA and a new algorithm P-HPartition as solutions \cite{acs2012differentially}. While both of these methods improve on the results of FPA for histogram publication, both explicitly take a ``true'' histogram as input complete with defined bins. Another more general method often extended to histogram publication is DPSense \cite{day2015differentially}. The authors present this algorithm explicitly as a method for the answering the query for a vector of column counts, which in themselves encode previous assumptions about the domain. 

Lu et al. consider generating synthetic data for testing untrusted systems \cite{lu:14}, but assume a matrix structure to the data that implies a known domain. Similarly, Xiao et al. consider synthetic data release through multidimensional partitioning, but use data cubes to explicitly map to a bounded n-dimensional space. There is also work that extends these ideas to correlated attributes and graphical models, most famously the PrivBayes algorithm \cite{chen14}. However, this also relies on a bounded domain in order to draw inferences about the correlation structure. The approach we adopt to account for large uncertain domain sizes is closely related to the sparse vector technique proposed by Cormode et al. \cite{cormode2011differentially}, which explicitly models elements beyond the active domain observed in the dataset.

\ignore{
As a simple example from \cite{pearl2001direct}, gender should not directly influence the hiring decision, but may influence the hiring decision indirectly, via  applicant qualifications.}

\section{Model and Algorithm}
\label{sec:method}

In this section, we model the bias repair problem, provide some background on causality, and present our solution.
We interpret the problem of removing bias from a dataset as eliminating a \emph{causal dependency} between a \emph{treatment attribute} $X$ and an \emph{outcome attribute} $Y$, assuming \emph{sufficient covariates} $\mb Z$. 

$X$ and $Y$ are conditionally independent given $\mb Z$ in $R$, written $( X \indep  Y| \mb Z)$, if

$$ \pr_{R}( X , Y, \mb Z )=\pr_{R}( X, \mb Z) \pr_{R}( Y| \mb Z)$$

The strength of a causal link between $X$ and $Y$ is measured by the conditional mutual information between $X$ and $Y$ given $\mb Z$ \cite{salimi2018bias}. It holds that $( X \indep  Y| \mb Z)$ iff the conditional mutual information between $X$ and $Y$ given $\mb Z$ is zero, written $I(X;Y| \mb Z)$.  To remove bias is to enforce $(X \!\perp\!\!\!\perp Y|  \mb Z)$ or, equivalently, to set the conditional mutual information between the treatment and the outcome given the sufficient covariates to zero.

Following an example from the introduction, we can consider the effect of bike share company on gender: one company may market more aggressively to women, or their bikes may be more difficult for men to ride.  This causal dependency warrants removal in various situations. For instance, the company may not want to reveal their marketing strategy, they may not want to be seen as propagating a gender bias, or a model trained on these results may be less generalizable to other cities if this bias is propagated.

\paragraph{\textbf{Problem Statement: Bias Repair}}
{\em Given a relation $R$ with a causal dependency $(X \not\!\perp\!\!\!\perp Y| \mb Z)$, and given a dissimilarity measure $\Delta$ between two probability distributions, the \emph{bias elimination problem} is to find $R'$ such that $(X \indep Y| \mb Z)$ while minimizing $\Delta(R,R')$.}

The dissimilarity measure $\Delta$ is interpreted as between  $\pr_{R}(\att)$ and $\pr_{R'}(\att)$ (e.g., KL-divergence). We will consider two different distance metrics in Section \ref{sec:metrics}: Position-weighted Kendall's Tau (which is rank-sensitive) and Hellinger distance (which is not). We defer a theoretical study of this optimization problem to our ongoing and future work, though we point out a connection to the problem of low-rank matrix approximation \cite{markovsky2011low}.  In this paper, we propose an algorithm that directly enforces the independence condition, then show experimentally that the effect on distance is small.

\subsection{Background on Causality}
\label{sec:method:causality}

We now briefly review causal inference, which forms the basis of our repair algorithm. 
%
The goal of causal inference is to estimate the effect of a {\em treatment} attribute $X$ on an {\em outcome} attribute $Y$ while accounting for the effects of covariate attributes $\mb Z$. We compute  a {\em potential outcome} $Y(x)$\cite{rubin2005causal}, which represents the outcome if, in a hypothetical intervention, the value of $X$ were set to value $x$.  The \emph{causal effect of $X$ on $Y$} is the expected value of the difference in the potential outcomes for two different values of $X$: $E[Y(x_1)-Y(x_0)]$, called the
{\em average treatment effect (ATE)}.

ATE can be computed if we can assume that a) missing attributes can be treated as having values that are effectively assigned at random (unconfoundedness/ignorability), and that b) it is possible to observe both positive and negative examples of $X$ in a relevant subset of the data (overlap).  These two conditions can be formalized as assuming a subset of attributes $\mb Z \subseteq \att$ is available such that: 
\begin{align*}
\forall \mb{z} \in &Dom(\mb Z),\\
&{Y(x_0), Y(x_1) \bigCI X} \; | \; \mb{Z}=\mb{z}   \quad &\text{(Unconfoundedness)}\\
& 0 < \Pr(X = x_1 \;  | \; \mb{Z}=\mb{z}) < 1 \quad &\text{(Overlap)} 
\end{align*}

If these conditions are met, ATE can be computed as follows:
%
\begin{eqnarray} \small
  \ate =  \sum_{\mb z \in Dom(\mb Z)} (\E[Y|X=x_1,\mb Z= \mb z]-\E[Y|X=x_0,\mb Z= \mb z]) \ \Pr(\mb Z= \mb z) \label{eq:adj}
\end{eqnarray}

\noindent where $Dom( \mb Z)$ is the domain of the attributes $ \mb Z$.

From this expression, it can be shown that the $\ate$ of $X$ on $Y$ is zero iff $I(X; Y| \mb Z)=0$. Therefore, we can use the conditional mutual information $I(X;Y| \mb Z)$ to quantify the strength of a causal link between $X$ and $Y$ given $\mb Z$. 

ATE quantifies the \emph{total} effect of $X$ on $Y$, which can be separated into direct effects and indirect effects (those that are mediated through other attributes).  In this paper, we ignore this distinction, and leave generalizing the method to account for this distinction to future work.

\subsection{Algorithm}
\label{sec:alg}
We propose a simple algorithm to compute an approximate solution to our problem. The algorithm is based on the intuition that $(X \indep Y| \mb Z)$ holds in $R'$ iff the joint probability distribution $\Pr_{R'}(\att)$ admits the following factorization, based on the chain rule:

\begin{align}
\pr_{R'}(\att)= \pr_{R'}(X \mb Z) \pr_{R'}(Y|\mb Z) \pr_{R'}(\mb U|XY \mb Z)  \label{eq:fac}
\end{align}

\noindent  where $\mb U = A-XY \mb Z$.  

This factorization will form the basis of our algorithm, but there is a complication: We want to restrict $R'$ to include only the active domain of $R$ rather than the full domain.  The reason is that transportation datasets are typically sparse; there are many combinations of attributes that do not correspond to any traffic (e.g., bike rides from the far North of the city to the far South).  We assume $R$ is a bag; it may contain duplicates.  For example, there may be multiple trips with the same origin, destination, and demographic information.  Under this semantics, we express our algorithm in terms of \emph{contingency tables}.

A contingency table over a set of attributes $\mb X \subseteq \att$, written $\mc{C}^{\mb X}_R$, is simply the count of the number of tuples for each unique value of $\sx \in Dom(\bx)$.  That is, $\mc{C}^{\mb X}_R$ corresponds to the result of the query \texttt{select $\mb X$, count(*) from $R$ group by $\mb X$}.  More formally, a contingency table over $\mb X \subseteq \att$ is a function $Dom(\mb X) \rightarrow \mathbb{N}$ 

$$\mc{C}^{\mb X}_R(\mb x)=\sum_{t \in R} \one[t[\mb X]= \mb x]$$
\noindent $t[\mb X]$ represents the tuple $t$ projected to the attributes $\mb X$, and $\one$ is the indicator function for the condition $t[\mb X]= \mb x$. 
The contingency table over all attributes in $R$ is an alternative representation for the bag $R$ itself: Given $\mc{C}^{\att }_R$, we can recover $R$ by iterating over $Dom(\att)$.  In practice, this step is not necessary, as $\mc{C}$ is implemented as a $k$-dimensional array.

Using contingency tables, we can compute a new joint probability distribution over $\att$ as
$$\pr_R(\att=a)=\frac{ \mc{C}^{\att}_R(a)}{|R|}$$


Algorithm \ref{algo:smi} uses these ideas to construct the desired relation $R'$ from the marginal frequencies of $R$, enforcing Equation \ref{eq:fac} by construction.  It can be shown that the KL-divergence between  $\pr_{R}(\att)$ and $\pr_{R'}(\att)$ is bounded by $I(X; Y|\mb Z)$. That is, the divergence of $R'$ from $R$ depends on the strength of the causal dependency between $X$ and $Y$.  If the causal dependency is weak, Algorithm \ref{algo:smi} will have no significant effect on the dataset.  We will evaluate the effects experimentally in Section \ref{sec:experiments}.

\begin{algorithm} 
	\DontPrintSemicolon
	\KwIn{ An instance $R$ with $\att= X Y \mb Z \mb U$ in which $(X \not\!\perp\!\!\!\perp Y| \mb Z)$}
	\KwOut{An instance  $R'$ in which   $(X \indep Y| \mb Z)$}
    
    \For{$x y \mb z \mb u \in R$} {
        $numerator \gets \mc{C}^{X \mb Z}_{R}(x \mb z) \mc{C}^{Y \mb Z}_{R}(y \mb z) \mc{C}^{X  Y \mb Z  \mb U}_{R}(x y \mb z \mb u)$\\
        $denominator \gets |R| \mc{C}^{\mb Z}_{R}(\mb z) \mc{C}^{X Y \mb Z}_{R}(x y \mb z)$\\
		$\mc{C}^{\att}_{R'}(xy \mb z \mb u) \gets \mb{Round}(\frac{numerator}{denominator})    $ \;
}

	$ \Return  \ \text{$R'$ associated with} \  \mc{C}^{\mb \att}_{R'}$
	
		\caption{Enforcing Conditional Independence} \label{algo:smi}
\end{algorithm}



\section{Applications and Datasets}
\label{sec:app}
\label{sec:taxi}

In this section we describe two real datasets to which we apply our methodology and an overview of how both datasets were processed for use in our evaluation.


\paragraph{NYC Taxi Data}
\label{sec:data:taxi}
The NYC taxi trip record dataset released by the Taxi \& Limousine Commission (TLC)~\cite{tlc} contains trips for 13,260 taxi drivers during January 2013, with pick-up and drop-off location as (lat,lon) coordinates and other information including trip distance and tip amount.  We used this particular release of the data because medallion numbers were no longer made available after this release.  
We first removed  transportation records with missing values, such as records with unknown pick-up or drop-off locations or missing tip amount. 
We then categorized trip distance into low, medium, and high, with about 1/3 of the trips falling into each category. Tip amount was categorized into low and high, with high tip corresponding to at least 20\% of the fare amount. Note that the original dataset has tip amount information only for rides that were paid by a credit card, and so we only consider these trips in the paper. Lastly, drivers were categorized into low, medium, and high frequency drivers. Table~\ref{tab:taxi_data} shows an example of the data after aggregation, with the count of each instance represented in the \texttt{count} column.

\begin{table}[htbp]
\centering
\caption{Processed NYC taxi data}
\label{tab:taxi_data}
\begin{tabular}{ccccccccc}
\toprule
o\_lon &  o\_lat &  d\_lon &  d\_lat & pickup & dist &   tip & freq & count\\
\midrule
 -74.0 & 40.7 & -74.0 & 40.7 & night & med & high & low & 6074 \\
 -74.0 & 40.7 & -74.0 & 40.7 & night & med &  low & low & 2844\\
 -73.9 & 40.7 & -73.9 & 40.7 & day   & low & high & med & 16 \\
 -73.9 & 40.7 & -73.9 & 40.7 & morn  & low & high & low & 14 \\
 -73.9 & 40.7 & -74.0 & 40.7 & morn  & low & high & high & 3 \\
\bottomrule
\end{tabular}
\end{table}

\paragraph{Dockless Bikeshare}
\label{sec:data:bike}
The bike data includes rides from $197,049$ distinct riders between June 2017 and May 2018 across three different companies. Each rider is identified via a unique rider id for each company, and the start and end location of each trip is projected to one of 94 neighborhoods in the Seattle area. 
Trip information is joined with rider information from survey responses, indicating their gender and whether or not they use a helmet.

\paragraph{Data Processing and Aggregation}
\label{sec:proc}
We pre-processed both datasets to make them compatible with our approach.  First, the time in both is precise up to the second. Since our model assumes categorical attributes, we map time to four buckets: morning (5am - 9am), day (9am - 3pm), evening (3pm - 7pm), and night (7pm - 5am).  Additionally, each individual driver/rider was classified into one of three categories by the number of trips they made, as recorded in the dataset. The top 1/3, who made the most trips, are designated \texttt{heavy}, the bottom 1/3 are designated \texttt{light}, and the rest are designated \texttt{medium}.

\section{Experiments}

\label{sec:experiments}
In this section, we first outline our evaluation metrics, and then present experiments to consider whether the error introduced by our bias-repair method is comparable to the error introduced by natural variation.  Recall that we wish to remove the causal dependency between $X$ and $Y$. If there is no correlation between these attributes, then the repair process will not change the weights significantly.  If there is a strong correlation, then the process will force the mutual information to zero while preserving the distribution of the other attributes.

We consider three  situations: synthetic data simulating extreme situations (Section \ref{sec:synthetic}), real datasets representing bike and taxi data (Section \ref{sec:realdata}), and the same real bike and taxi data, but aggregated post hoc to simple origin-destination pairs (Section \ref{sec:aggregated}).
The experiments in each of these situations can be summarized by the choice of treatment ($X$), outcome ($Y$) and covariate ($Z$) attributes, $X \rightarrow Y | Z$, as follows:

\begin{enumerate}
	\item Synthetic: $gender \rightarrow rating | \{origin,destination\}$
    \item Bike: $company \rightarrow gender | \{start\_nhood, end\_nhood,time\_of\_day,helmet\}$
    \item Taxi: $distance \rightarrow tip | \{orig\_lon, orig\_lat, dest\_lon, dest\_lat\}$
\end{enumerate}



\subsection{Evaluation Metrics}
\label{sec:metrics}
Our goal is to remove the effect of the given relationship without destroying the utility of the resulting dataset.  The proposed method would not be viable if it altered the distribution of traffic ``too much.''  To define ``too much,'' we a) compute the distance between the original dataset and the adjusted dataset, and b) compare this distance with the distances associated with a set of bootstrap samples of the original dataset.  If the distance with the adjusted dataset falls within the distribution of the bootstrap samples, we conclude that the adjustment is small enough to still produce a useful dataset.

To compute distances, we consider two different metrics: one that is rank-sensitive, and one that is not.  To measure rank-sensitive distance, we sort the buckets by trip count in descending order before and after the repair.  We then use position-weighted Kendall's tau \cite{DBLP:conf/www/KumarV10} to compare the two resulting rankings.  Kendall's tau counts the number of pair-wise position swaps between a ground truth ranking and an experimental ranking.  Position-weighted Kendall's tau incorporates a weighting function, usually to assign more importance to swaps that happen closer to the beginning of the ranked list.\footnote{Many methods for comparing ranked lists have been proposed.  We opt for a measure in which identity of the items being ranked (histogram buckets) is deemed important.  This is in contrast to typical IR measures such as NDCG or MAP, where item identity is disregarded, and only item quality or relevance scores are retained.}  This measure is appropriate in our domain, because a) transportation analysts and engineers are primarily interested in the conditions associated with the heaviest traffic flows, and b) transportation datasets are inherently sparse. 

The weighting function we consider is harmonic: Given position $i$ in a ranking, the weight is $\frac{1}{i}$.  We also considered an exponential weighting function, since traffic patterns tend to follow an exponential distribution, but that weighting function was potentially too generous to our method: The first few positions were all that mattered.

To measure distance independently of rank and position, we use Hellinger distance. This measure is an f-divergence closely related to the Bhattacharyya distance that obeys the triangle inequality, and is defined as follows: Let $p,q$ be two probability distributions over the same set of attributes $\bx$, and define the Bhattacharyya Coefficient $BC(p,q)$ to be $\sum_{x \in \bx}\sqrt{p(x)q(x)}$. Then the Hellinger distance is $H(p,q) = \sqrt{1-BC(p,q)}$.


\begin{table}
\centering
	\caption{Results of evaluation metrics across all experiments}
    \begin{tabular}{ c || c | c | c | c | c | c}
    \textbf{PWKT } & Synth. & Synth. & Bike & Bike & Taxi & Taxi \\
    & uncorrolated & correlated & all & grouped & all & grouped \\ \hline \hline
    2.5\% & 1.47 & 1.35 & 1.49 & 1.37 & 0.84 & 0.27 \\ \hline
    Mean & 2.93 & 2.44 & 2.34 & 2.88 & 1.39 & 0.81 \\ \hline
    97.5\% & 4.39 & 3.54 & 3.18 & 4.39 & 1.93 & 1.36 \\ \hline
    \textbf{Result} & \textbf{0.159} & \textbf{3.18} & \textbf{1.53} & \textbf{1.21} & \textbf{1.37} & \textbf{0.40} \\ \hline \hline
    \textbf{Hellinger} & Synth. & Synth. & Bike & Bike & Taxi & Taxi \\
    & uncorrolated & correlated & all & grouped & all & grouped \\ \hline \hline
    2.5\% & 0.075 & 0.072 & 0.084 & 0.14 & 0.029 & 0.044 \\ \hline
    Mean & 0.076 & 0.073 & 0.085 & 0.15 & 0.030 & 0.051 \\ \hline
    97.5\% & 0.076 & 0.074 & 0.086 & 0.15 & 0.030 & 0.047 \\ \hline
    \textbf{Result} & \textbf{0.00079} & \textbf{0.42} & \textbf{0.15} & \textbf{0.042} & \textbf{0.024} & \textbf{0.0020} \\\hline
	\end{tabular}
    \label{tab:results}
\end{table}

Table \ref{tab:results} presents results for both position-weighted Kendall's tau (PWKT) and Hellinger distance in each of our experiments.  The experimental result for Algorithm \ref{algo:smi} is in bold, and the other columns summarize the distribution of the bootstrap samples.  Figure \ref{fig:results} visualizes these results. Each experiment is represented by three bars. The light bar on the left shows the distribution of distances from the bootstrap procedure: the top of the bar represents the 97.5 percentile, the next line represents the mean, and lowest line represents the 2.5 percentile.  We visualize the distribution as a bar to emphasize that the measure is a distance, such that a lower bar is always better. The dark bar in the center is the experimental result. The final bar on the right represents a baseline test of assigning every trip a random $X$ value as a strategy of enforcing $I(X;Y)= 0$.

Overall, we can see that the error introduced by our algorithm is usually significantly less than the error one can expect from sampling, suggesting that the method is viable for correcting bias while retaining utility.

The expected variation is clearly visible for the case of PWKT, but for the Hellinger distance it is small compared to the magnitude of the metric, and is nearly impossible to distinguish precisely. The \emph{Correlated} and \emph{Bike} columns for the Hellinger distance stand out as significant outliers.

\begin{figure}

    \includegraphics[width=2.5in]{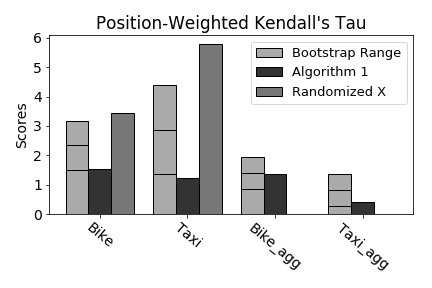}
	\includegraphics[width=2.5in]{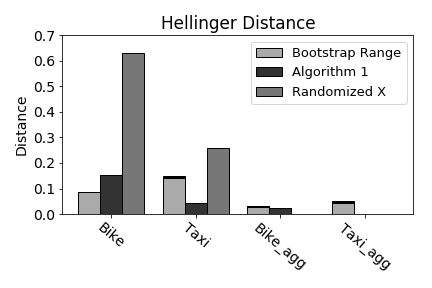}
    \caption{\label{fig:results}Expected bootstrap variation (left), experiment outcome (center), and baseline comparison (right) for each of the four experiments on real data. Lines in the bar denote the $2.5\%$, mean, and $97.5\%$ values for the distribution.}
\end{figure}

\subsection{Synthetic Ride Hailing Data}
\label{sec:synthetic}
For the synthetic experiments, the task is to remove the causal influence of gender on rating, simulating the situation where a data publisher does not want to unintentionally encourage discrimination \cite{NBERw22776}.

To generate the synthetic data, we use neighborhood-level trip data from the dockless bikeshare to simulate a realistic distribution of traffic among neighborhoods. Then, we assign each individual trip a gender at random from $\{m, f,o\}$ representing male, female, or other.
The \emph{no correlation} experiment assigns ratings according to a pre-defined distribution independent of gender, while the \emph{gender correlated} experiment uses three different distributions, one for each gender value, to simulate a strong correlation.  In both of these experiments, our simulated repair is to remove the effect of gender on rating, conditional on the origin and destination neighborhoods. 

We expect that the uncorrelated case should have minimal effect on the data, since there is no causal dependency to eliminate.  For the strongly correlated case, we expect the error to be significant.

Comparing the synthetic data experiments in Table \ref{tab:results}, we see that there is a change in the order of magnitude of the effect when the repair is acting on a relationship with a strong underlying correlation. When applied to synthetic data with no correlation structure, we find that values of both position-weighted Kendall's tau and Hellinger distance fall well below the range of error introduced by bootstrap sampling. However, in the correlated case when there was in fact a strong relationship, position-weighted Kendall's tau jumped to the upper extreme of the bootstrap range, and Hellinger distance far exceeded this range. This result indicates that the repair is causing a more drastic change in the gender correlated case than in the case with no correlation, as we would expect. 

Position-weighted Kendall's tau still falls within the bootstrap variation for the correlated case, which can be explained by the fact that certain neighborhood origin-destination pairs carry a disproportionate amount of the traffic in the dataset, so this relationship is preserved. The full magnitude of the change is better observed through the Hellinger distance in Table \ref{tab:results}, which grows an order of magnitude beyond the bootstrap variance in the gender correlated experiment.

\subsection{Real-World Bike and Taxi Datasets}
\label{sec:realdata}
In the bike experiment we remove the influence of company on gender using the dockless bikeshare data described in Section \ref{sec:app}. In this experiment, we are considering the situation where companies are releasing data to support traffic research, but do not want to expose any latent gender bias that may be attributable more to marketing efforts than to sexism.  The relationship between company and gender is conditional on origin, destination, and whether or not the rider uses a helmet.  In other words, \emph{only the effect of company should be removed, not the overall pattern of gender on ridership.}

In the taxi experiment we investigate the effect of a repair on the taxi data from Section \ref{sec:app}, in which we remove the influence of distance on tip amount, conditional on origin and destination.  The situation we consider is a behavioral economic analysis of tipping patterns, but we want to completely remove the influence of distance. Simply normalizing by distance is not enough, as the joint distribution between, say, time of day, distance, and tip amount can be complex.  Moreover, certain neighborhood origins and destinations may generate higher tips or lower tips in ways that interact with distance traveled.  For example, long east-west trips at certain times may be relatively short, but generate higher tips.

In both cases, we see that the calculated position-weighted Kendall's tau and Hellinger distance in Table \ref{tab:results} fall close to the expected variation from bootstrap samples, with the exception of the Hellinger distance for the bike share data, which is about twice this baseline. This anomaly helped us discover a data ingest error upstream from our algorithm: gender information was only properly included for one company, while the other two had two different default values. As a result, there was an unrealistically high correlation between company and gender. The order of values was still largely preserved by Algorithm \ref{algo:smi}, as seen in Figure \ref{fig:results}, since there are significantly more trips from one company than from the others, but the structural change results in a high Hellinger distance. Taken along with the taxi data, this reaffirms that Algorithm \ref{algo:smi} behaves as expected: it induces larger changes when there is a high degree of correlation in the relationship chosen for treatment.

\subsection{Aggregated Origin-Destination Data}
\label{sec:aggregated}

In our experiments so far we considered all possible fine-grained buckets in the dataset. For example, the trip count associated with \texttt{\{UDistrict, Downtown, Female, Helmet, Morning\}} appears as a bucket.  We also consider a coarser aggregated view of this data, grouping buckets by origin and destination and aggregating over gender, helmet, and time.  The motivation is that, in many situations, only origin-destination counts are important, and also that our method may unfairly benefit on a fine-grained dataset: if we preserve the distribution of the top few origin-destination pairs, we will also preserve the distribution of a large number of finer-grained buckets that divide these origin-destination pairs by gender, helmet and time. We run the same experiments and metrics as before, but this time grouping by origin and destination.

When aggregating as described, we see in Figure \ref{fig:results} that the baseline (right column) for each of these experiments has a value of $0$. This is because origin and destination were not included in $X$ or $Y$, and any repair that only takes into account the relationship between $X$ and $Y$ does not impact the other direct relationships in the dataset. For the results of Algorithm \ref{algo:smi} (center), the Hellinger distance falls below the expected variation for both datasets, while the position-weighted Kendall's tau falls in the bottom half of the expected range of variation. We therefore conclude that Algorithm \ref{algo:smi} preserves both order and structure of real aggregated data at least as well as a bootstrapped sample, given these particular correlation structures.

\section{Interactions with Privacy}
\label{sec:priv}

In this section we explore how approaches in differential privacy interact with our bias-reduction algorithm. Our goal is to create an algorithm that interfaces easily with algorithm \ref{algo:smi} while also accounting for the fact that transportation datasets are often very sparse in the sense that there are often many combinations of attributes that do not appear in the real data, even while others are frequent (e.g. bus rides taken between neighborhoods that are not connected along a bus route). Accounting for sparsity is a key challenge, since common privacy techniques applied to  large, sparse domains can destroy utility \cite{cormode2011differentially}.

Fundamentally, aggregation techniques to ensure privacy are insufficient, especially in a transportation context. Montoye et al. showed that just four points of trajectory data are sufficient to uniquely identify most users \cite{demontjoye2013unique}.  Therefore, differential privacy techniques, where noise is added to prevent inference about individuals based on the principle of indistinguishability \cite{dwork2006differential} are preferred. 

\textbf{Preliminaries.}
Two database instances $D_1$ and $D_2$ are \textit{neighbors} if they differ by exactly one row. Next, we imagine some randomized function $q$ and a set of results $S \subseteq Range(q)$. 

\begin{definition}[Differential Privacy] \label{dpdef}
A randomized function $q$ is $\epsilon$-differentially private if for all neighboring database instances $D_1$ and $D_2$ and all $S \subseteq Range(q)$,
\begin{equation}
Pr[q(D_1) \in S] \leq e^{\epsilon} \times Pr[q(D_2) \in S]
\end{equation}
where the probability is taken over the randomness of $q$.
\end{definition}

Definition \ref{dpdef} guarantees that the two neighboring database instances are indistinguishable to within a factor of $e^{\epsilon}$ when presented with a result of the randomized query $q$. This mechanism $q$ addresses concerns that a participant might have about leakage of their information through the results of $q$. Even if the participant decided not to include their data in the dataset, no outputs of $q$ would become significantly more or less likely.

Real datasets like those in transportation use cases often run into the issue of sparse domains. For example, in the dockless bikeshare data the combination of all possible attribute values yields over 2.5 million possible distinct bins. Since many trips fall into the same bins, even with orders of magnitude more trips than bins we would expect to have a significant number of bins with no trips. By extension, this means that we need to be careful about how we reason about the bins outside of the active domain but still admissible as part of the global domain. We formalize this well-known problem as follows, and use the formalization to provide a proof that there is no way to avoid including bins from outside the active domain:

\begin{lemma}
\label{zeroprob}
Given a differentially private mechanism 
$q$ with respect to a set of database instances $\mathcal{D}$ of schema $R(\mathcal{A})$, if $P[q(D_i) = r] = 0$ for some database instance $D_i$ and a result $r$, then $P[q(D) = r] = 0$ for all $D \in \mathcal{D}$.
\end{lemma}
\begin{proof}
Let $P[q(D_i) = r] = 0$ for some database instance $D_i$ and result $r$. For any database instance $D_j \in \mathcal{D}$, define a sequence of database instances $S(D_j, D_i) = \{D_1, D_2, \ldots , D_n\}$ such that $D_1 = D_i$, $D_n = D_j$, and $D_k$ and $D_{k+1}$ are neighbors for all $0 \leq k < n$. \footnote{Two datasets are neighbors if they differ in the presence or absence of a single record, following the differential privacy definition.} It follows that
\begin{align*}
	P[q(D_j) = r] &\leq e^{\epsilon} \times P[q(D_2) = r] \\
    &\leq e^{2\epsilon} \times P[q(D_3) = r] \\
    &\vdots \\
    &\leq e^{(n-1)\epsilon} \times P[q(D_i) = r] \\
    &\leq 0
\end{align*}
Since the probability cannot be less than $0$, we find that $P[q(D) = r] = 0$ for all $D \in \mathcal{D}$
\end{proof}

\begin{theorem}
\label{noactive}
There does not exist a differentially private mechanism that \emph{only} returns elements from the active domain.
\end{theorem}

The intuition for Theorem \ref{noactive} is that every element in the global domain must have a non-zero chance of being included in the result, or else Lemma \ref{zeroprob} is violated.

\begin{proof}
Let $\mathcal{D}$ be the set of possible database instances. Let $c$ be a function such that $c(D)$ returns the active domain of $D$. Let $c^*$ be some differentially private mechanism that returns a value set $c^*(D) \subseteq c(D)$.  This means that $c^*$ can be probabilistic, but only the value sets in the powerset of $c(D)$ have non-zero probabilities. Next let us pick some $D_i \in \mathcal{D}$. Then $P[q(D_i) = r] = 0$ for all $r \not\subseteq c(D_i)$, and additionally by Lemma \ref{zeroprob} that $P[q(D) = r] = 0$ for all $D \in \mathcal{D}$. This holds true for all $c(D)$, and thus the possible range $R$ of $c^*(D)$ is constrained by $R = \cap_{D \in \mathcal{D}} c(D) = \emptyset$. Therefore, our mechanism $c^*$ must be the trivial one that returns the empty set.
\end{proof}

Theorem \ref{noactive} depends crucially on the choice of function $c$. If we allow this function to be defined on the global domain of the database instance $D$, we can in fact create a differentially private mechanism. This observation is the motivation behind the approach that we take in Algorithm \ref{algo:cat_hist}, which is closely related to the technique outlined by Cormode et al. \cite{cormode2011differentially}. Our approach differs in that we parameterize our algorithm by the \emph{tolerance for including values from outside of the active domain} rather than as a fixed number of such values.
This formulation is essentially a usability enhancement: users cannot necessarily provide the number of bins (or a threshold defining the number of bins) without inspecting the dataset, while a probability for including bins can be estimated globally based on the use case: A heatmap of trips may tolerate a few trips in unexpected places, while an analysis of maximum trip distance or other computations that are sensitive to outliers have a lower tolerance. We also add laplacian noise rather than geometric noise \cite{ghosh2012universally} for simplicity.

Our algorithm is differentially private in the context of a global domain, and controls the number of bins from outside the active domain by deriving a threshold from the user-provided tolerance $\rho$. Reducing $\rho$ will lead to choosing a higher threshold and therefore excluding more of the active domain, but any choice for $\rho$ fulfills the criteria for differential privacy.  In a sense, this formulation frames the question a one of data sufficiency: Given a tolerance for irrelevant bins, do you have enough data to include the most important bins in the histogram with high probability?  If not, you can still release the dataset (and retain differential privacy guarantees at the expense of utility), or you can potentially go and collect more data to produce higher bin counts and improve their likelihood of being included.

\begin{algorithm} 
	\DontPrintSemicolon
	\KwIn{Difference in sizes of global and active domains $n$, tolerance for values outside the active domain $\rho$, privacy budget $\epsilon$, and the true histogram ($C$, $S$) where $C$ is a vector of categories in the active domain and $S$ is a vector of true frequencies for each corresponding category in $C$.}
	\KwOut{Differentially private histogram}
	
	$\tau \gets -\large\frac{\ln(2(1-\rho^{\frac{1}{n}}))}{\epsilon}$\\
	$i \gets 0$ \\
	\While{$i < |C|$}{
	    $s_i \gets $LaplaceDistribution$(s_i, \frac{1}{\epsilon})$ \\
      	\If{$s_i < \tau$}{Remove($c_i, s_i$)}
        $i \gets i+1$ 
	} 
    $k \gets $BinomialDistribution$(n, \frac{1}{2}e^{-\epsilon \tau})$ \\
    $j \gets 0$\\
	\While{$j < k$}{
      	Append($C$, GetCategoryFromDomain())\\
        Append($S$, $\tau + $ ExponentialDistribution$(\frac{1}{\epsilon})$)\\
      	$j \gets j+1$\\
    }
    \Return{$C,S$}
	
	\caption{Categorical Histogram Method}\label{algo:cat_hist}
\end{algorithm}

To see how this works in practice, we run algorithm \ref{algo:smi} with \ref{algo:cat_hist} for various values of $\rho$ and $\epsilon$. We can then compare these results to both the output of algorithm \ref{algo:smi}. A key finding is that for these datasets our algorithms can be applied in either order: privacy-first, or bias-first. In other words we can compose the bias reduction and noise injection steps in either order. However, the privacy guarantee is subtly different between these two cases: If we remove bias first, the distribution is a sample from the ``fair'' world where biases have been removed, which is different than the ``real'' world. Specifically, a record in the bias-corrected dataset does not necessarily correspond to any real individual.  Therefore, we present results from the privacy-first approach in order to make the privacy guarantee easier to interpret.


\begin{figure}
    \centering
    \includegraphics[width=2.3in]{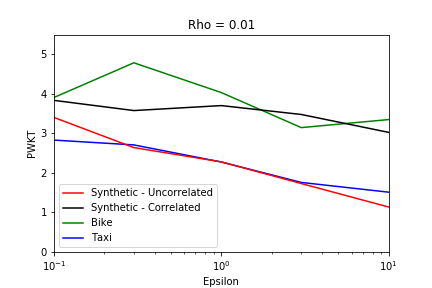}
    \includegraphics[width=2.3in]{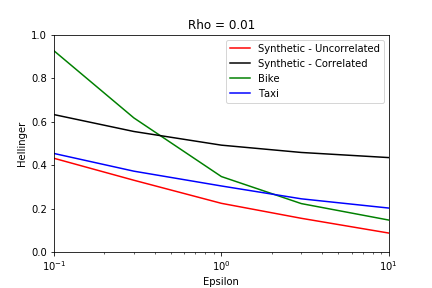}
    \includegraphics[width=2.3in]{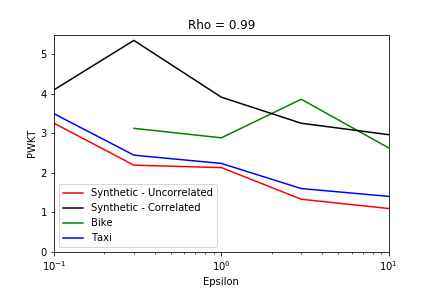}
    \includegraphics[width=2.3in]{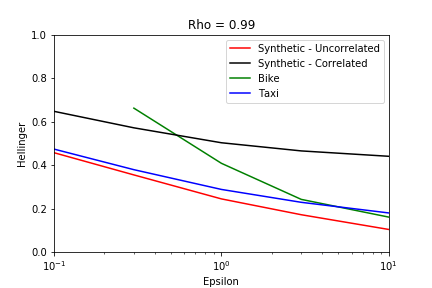}
    \caption{Distance metrics of results of composing algorithms \ref{algo:smi} and \ref{algo:cat_hist} over all four datasets for varying values of $\rho$ and $\epsilon$.}
    \label{fig:noise_results}
\end{figure}

Figure \ref{fig:noise_results} shows the results over the same four datasets and attributes that we investigated in section \ref{sec:experiments} for ease of direct comparison. For the same two evaluation metrics as in the previous section we show a plot of the metric vs $\epsilon$ on a logarithmic scale for a large and small value of $\rho$, where each data point used comes from the average over $10$ trials in order to account for the inherent randomness of the algorithm. 

For some very conservative parameter combinations, the derived threshold was higher than all buckets, leading to an empty result. For example, $\epsilon = 0.1$ and $\rho = 0.99$ in the Taxi dataset implies a minimum bin count of $245$. This threshold exceeded the support of any single bucket in the histogram, and our algorithm \ref{algo:cat_hist} returned no data for which to calculate distances. Empty results like these demonstrate how algorithm \ref{algo:cat_hist} yields differentially private results no matter what the input parameters are. In practice, the data owner could potentially choose to postpone data release until more data is collected, to increase the buckets' support and eventually produce non-empty output.

Next, we observe that the Hellinger distances (at right) are higher than those without noise, but converge to the noiseless values as the privacy budget $\epsilon$ is increased. The rate at which they converge is directly related to the size of the global domain. Since the two synthetic datasets were created using the same set of possible origin and destination neighborhoods as the bike data, the size of the global domains is similar across all three datasets. On the other hand, the taxi data is sampled from a much smaller global domain. Therefore, increasing $\epsilon$ from $0.1$ to $1$ reduces the Hellinger distance by about $0.1$ for the bike and synthetic datasets, but $0.5$ for the taxi dataset. For low values of $\epsilon$, the taxi results are dominated by bins included from the global domain, yielding very high distances.

However, PWKT (at left) behaves a bit less smoothly. There does still seem to be an overall decreasing trend as $\epsilon$ increases as we would expect, but this does not happen monotonically. The spikes that occur along the way come as a result of the linear decrease of the threshold chosen by \ref{algo:cat_hist} that comes with increasing $\epsilon$. In particular, since PWKT is rank-sensitive, the distance it calculates depends on both the number of ranked buckets as well as their order. Thus having a high threshold that yields only a few buckets for analysis will tend to result in a smaller distance value, even if the buckets that were preserved are not preserved in order. This causes an inflection point at a different $\epsilon$ value for each value of $\rho$, as eventually the distance-reducing effect of more accurate ordering counteracts the distance-increasing effect of including more buckets in the distance calculation. Domain size and these inflection points also help explain why the results from the taxi data seem to fail to converge to the no-noise case towards $\epsilon = 10$, even as the other three datasets do so.

\begin{figure}
    \centering
    \includegraphics[width=2.3in]{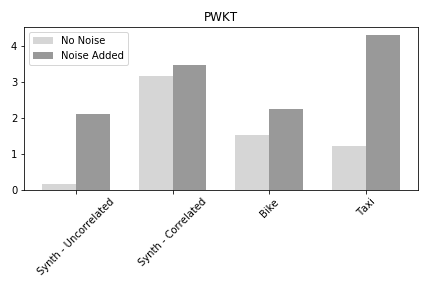}
    \includegraphics[width=2.3in]{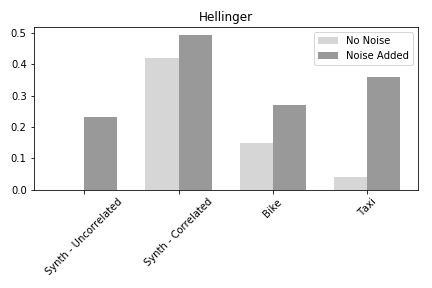}
    \caption{Comparison of the no noise case from section \ref{sec:experiments} (left) and noised results (right) for fixed $\rho = 0.5$ and $\epsilon = 1$.}
    \label{fig:noise_comparison}
\end{figure}

Figure \ref{fig:noise_comparison} compares the results of composing algorithms \ref{algo:smi} and \ref{algo:cat_hist} to the no noise case presented in section \ref{sec:experiments} for a reasonable choice of $\rho$ and $\epsilon$. Adding noise increases the distances across the board, and especially to the taxi data whose support is smaller relative to the global domain. However, we also observe that the increase in distance is larger when the distance due to bias correction was smaller. This result implies that the randomization induced by algorithm \ref{algo:cat_hist} overlaps with the bias reduction of algorithm \ref{algo:smi}. The relationship between these algorithms hints at a more fundamental relationship between the processes of bias reduction and privacy preservation, which is a potential direction for future work.

\section{Conclusions and Future Work}
\label{sec:conc}

Data sharing is emerging as a critical bottleneck in urban and social computing.  While risks associated with privacy have been well-studied, data owners and data publishers must also be selective about the patterns they reveal in shared data. Biases in the underlying data can be reinforced and amplified when used to train models, leading to not only poor quality results but also potentially illegal discrimination against protected groups, causing a breach of trust between government and companies.

In this paper, we have considered the bias-correction problem --- an important pre-processing step in releasing data that is orthogonal to privacy.

We interpret the need to repair unintended or unrepresentative relationships between variables prior to data release as related to causal inference: the conditional mutual information between two variables is a measure of the strength of the relationship.  We propose an algorithm that interprets the frequencies of trip events as a probability distribution, then manipulates this distribution to eliminate the unwanted causal relationship while preserving the other relationships.

We show that this procedure produces expected behavior for synthetic datasets representing extreme cases, and has only a modest impact in real datasets: the distance between the original data and the adjusted data falls within the bounds of natural variation of the original data itself. Additionally, we present a method to adjust our approach to meet the standards of differential privacy - a crucial step for the adjustment and release of any real-world dataset.

Going forward, we aim to generalize this approach to other domains, distinguish between direct and indirect causal effects, and explore new algorithms that can better balance the tradeoff between utility and causal relationships. We also hope to investigate the ways that this approach interacts with privacy-preserving methods in more detail, including direct overlap that could be leveraged to decrease the amount of noise added as well as formalizing what it might mean to preserve privacy after the bias reduction step. Our broader vision is to develop a new kind of open data system that can spur data science research by generating safe and useful synthetic datasets on demand for specific scenarios, using real data as input.

\bibliographystyle{abbrv}
\bibliography{refs,bigdata}  

\end{document}